\documentclass[12pt]{article}
\usepackage{}
\usepackage{amsthm,amsmath,indentfirst,amssymb}
\usepackage{algorithm,algpseudocode,color}
\usepackage{mathrsfs}
\usepackage{hyperref}
\usepackage{url}
\usepackage{pstricks}
\usepackage{tikz}
\usepackage{caption}
\usepackage{subfigure}
\newtheorem{lemma}{Lemma}[section]
\newtheorem{theorem}[lemma]{Theorem}

\newtheorem{definition}[lemma]{Definition}

\newtheorem{remark}[lemma]{Remark}

\begin{document}\baselineskip 0.56cm

\title{Approximation Algorithm for Minimum $p$ Union Under a Geometric Setting}
\author{\footnotesize Yingli Ran$^1$\thanks{Corresponding author: Yingli Ran, ranyingli@zjnu.edu.cn}\quad Zhao Zhang$^1$\thanks{Corresponding author: Zhao Zhang, hxhzz@sina.com}\\
	{\it\small $^1$ College of Mathematics and Computer Science, Zhejiang Normal University}\\
    {\it\small Jinhua, Zhejiang, 321004, China}}
\date{}
\maketitle

\begin{abstract}
In a minimum $p$ union problem (Min$p$U), given a hypergraph $G=(V,E)$ and an integer $p$, the goal is to find a set of $p$ hyperedges $E'\subseteq E$ such that the number of vertices covered by $E'$ (that is $|\bigcup_{e\in E'}e|$) is minimized. It was known that Min$p$U is at least as hard as the densest $k$-subgraph problem. A question is: how about the problem in some geometric settings? In this paper, we consider the unit square Min$p$U problem (Min$p$U-US) in which $V$ is a set of points on the plane, and each hyperedge of $E$ consists of a set of points in a unit square. A $(\frac{1}{1+\varepsilon},4)$-bicriteria approximation algorithm is presented, that is, the algorithm finds at least $\frac{p}{1+\varepsilon}$ unit squares covering at most $4opt$ points, where $opt$ is the optimal value for the Min$p$U-US instance (the minimum number of points that can be covered by $p$ unit squares).

\vskip 0.2cm\noindent {\bf Keyword}: minimum $p$ union; unit square; approximation algorithm.
\end{abstract}

\section{Introduction}\label{sec1}
The {\em minimum $p$ union problem} (Min$p$U) was first proposed by Chlamt\'{a}\v{c} et al. \cite{Chlamtac}. Given a hypergraph $G$ with vertex set $V$ and hyperedge set $E$, together with an integer $1\leq p\leq m$, where $m$ is the number of hyperedges, the goal of Min$p$U is to select $p$ hyperedges such that the number of vertices in their union is as small as possible. A $2\sqrt{m}$-approximation algorithm was given in \cite{Chlamtac}, which was further improved to $O(m^{1/4+\varepsilon})$ in \cite{Chlamtac1}. Min$p$U is a generalization of the {\em smallest $p$-edge subgraph problem} \cite{Dinitz} (S$p$ES) to hypergraphs, and S$p$ES is a dual problem of the {\em densest $k$-subgraph problem} (D$k$S). The hypergraph extension of D$k$S is the {\em Densest $k$-subhypergraph} problem (D$k$SH) \cite{Chlamtac}. Given a hypergraph $G=(V,E)$ and an integer $k$, the goal of D$k$SH is to find a vertex subset $V'\subseteq V$ of size at most $k$ that contains the largest number of hyperedges, where $V'$ contains a hypergraph $e$ means $e\subseteq V'$. In \cite{Chlamtac}, the authors proved that if there is an $f$-approximation for D$k$SH, then there is an $O(f\log p)$ approximation for Min$p$U.

Although there are a lot of researches on D$k$S and S$p$ES, researches on Min$p$U and D$k$SH are rare. In \cite{Chlamtac1}, Chlamt\'{a}\v{c} et al. pointed out that ``Given the interest in and importance of D$k$S and S$p$ES, it is somewhat surprising that there has been very little exploration of the equivalent problems in {\em hypergraphs}''. Besides, they also pointed out that ``it is widely believed that D$k$SH and Min$p$U do not admit better than a polynomial approximation ratio''. Then a natural question arises: can we obtain better approximation when considering geometric setting? In a geometric setting, the element set is a set of points on the plane, denoted as $V$. Given a set of objects $\mathcal S$ on the plane, the hyperedge set is $E=\{V\cap s\colon s\in\mathcal S\}$, where $V\cap s$ is the set of points contained in object $s\in\mathcal S$. In such a setting, the Min$p$U problem can be viewed as selecting $p$ objects such that the number of points covered by the union of these objects is as small as possible.

In this paper, we consider the geometric Min$p$U problem in which the objects are unit squares (denote this problem as Min$p$U-US), and obtain a $(\frac{1}{1+\varepsilon},4)$-bicriteria approximation algorithm, where an $(\alpha,\beta)$-bicriteria approximation algorithm means that the approximation ratio is $\beta$ and the feasibility is violated by a factor of $\alpha$.

Min$p$U-US has a background in the construction of obnoxious facilities, such as garbage collection stations. The point set $V$ corresponds to the locations of inhabitants. A garbage collection station has some obnoxious effect on those inhabitants in a unit square surrounding it. Suppose $p$ garbage collection stations are planned to be established. It is desired that the number of affected inhabitants is minimized.

\subsection{Related Works}

In 1993, Kortsarz et al. proposed the first approximation algorithm for the {\em densest $k$ subgraph problem} (D$k$S), which achieves approximation ratio $O(n^{2/5})$ \cite{Kortsarz}, where $n$ is the number of vertices. Currently, the best known approximation ratio for D$k$S is $O(n^{1/4+\varepsilon})$ \cite{Bhaskara}.
In \cite{Manurangsi}, it was proved that D$k$S cannot be approximated within factor $O(n^\frac{1}{(\log \log n)^c})$ for some constant $c$ under the ETH assumption and it is widely believed that D$k$S might not have subpolynomial approximation ratio \cite{Chlamtac1}. But in some special case, the approximation might be better. In \cite{Ashahiro}, Ashahiro et al. presented an $O(n/k)$-approximation for D$k$S by a simple greedy strategy. Feige and Langberg \cite{Langberg} showed that
the approximation ratio $n/k$ is achievable using a semidefinite programming. Note that when $k=\Omega(n)$, ratio $n/k$ is a constant. In the case when $k=\Omega(n)$ and the graph is dense (that is, the number of edges is $\Omega(n^2)$), Arora et al. \cite{Arora} presented a polynomial-time approximation scheme (PTAS) for D$k$S using a random sampling technique. Finding a D$k$S remains NP-hard even for chordal graphs and bipartite graphs \cite{Per}. A PTAS was presented for D$k$S in interval graphs by Nonner \cite{Nonner}.

Min$p$U is closely related to the {\em small set vertex expansion problem} (SSVE) \cite{Louis}. Given a graph $G=(V,E)$ and an integer $p$, the goal of SSVE is to select a subset $S\subseteq V$ with $p$ vertices such that $|N_{G}(S)|$ is minimized, where $N_{G}(S)=\{u\in V\setminus S\colon \exists v\in S~\mbox {s.t}~(u,v)\in E\}$ is the neighbor set of $S$. Min$p$U is equivalent to a bipartite SSVE problem (SSBVE) \cite{Chlamtac1}: given an instance of Min$p$U, we can construct a bipartite graph in which the left side represents hyperedges, the right side represents vertices, and there is an edge between a hyperedge and a vertex if the hyperedge contains the vertex. For general $p$, Chlamtac et al. \cite{Chlamtac1} obtained a $(1+\varepsilon,\tilde{O}(\sqrt{n}/\varepsilon))$-bicriteria approximation algorithm for SSVE.

Min$p$U is also closely related to the {\em minimum partial set multi-cover problem} (MinPSMC). Given a hypergraph $G=(V,E)$ and an integer $1\leq k\leq n$ where $n$ is the number of vertices, each vertex $v\in V$ has a covering requirement $r_v$, the goal of MinPSMC is to select the minimum number of hyperedges to fully cover at least $k$ vertices, where a vertex $v$ is fully covered if it belongs to at least $r_v$ selected hyperedges. Ran et al. were the first to study the MinPSMC problem \cite{Ran3}. It was shown that Min$p$U is a special case of the MinPSMC problem, and the MinPSMC problem is at least as hard as the D$k$S problem \cite{Ran2}. Because of this hardness result, Ran et al. \cite{Ran} began to study the MinPSMC problem in a geometric setting. They studied the {\em unit-sqaure} MinPSMC problem, in which every vertex corresponds to a point on the plane and every hyperedge contains those points in a unit square. A PTAS was obtained for the special case when $r_v$ equals the frequency of $v$, that is, $v$ is fully covered only when all those hyperedges containing $v$ are selected. Some terminologies and ideas in \cite{Ran} will be used in this paper. However, it should be noted that the problem we are now studying is much different from that in \cite{Ran}, new insights have to be explored and new techniques have to be developed.

\subsection{Our contributions}

Since the Min$p$U problem in a general setting is very difficult, we study the geometric Min$p$U problem in which hyperedges correspond to unit squares, and present a $(\frac{1}{1+\varepsilon},4)$-bicriteria approximation algorithm for Min$p$U-US, which means that the approximation ratio is a constant while the feasibility is violated by a constant.

We first reveal a relationship between Min$p$U and D$k$SH in terms of bicriteria performance: D$k$SH has a $(\beta,\alpha)$-bicriteria approximation algorithm implies that Min$p$U has an $(\alpha,\beta)$-bicriteria approximation algorithm. So, the study of Min$p$U is transformed into the study of D$k$SH. For D$k$SH-US, we present a $(4,\frac{1}{1+\varepsilon})$-bicriteria approximation algorithm employing a strategy of partition and shifting. The main part is to design an algorithm for a subproblem on a block of constant size. This is done through a dynamic programming over refined grids of the block. For this purpose, we have to ``guess'' the envelope sets  (which form the boundaries for the union of those squares in an optimal solution), as well as those squares completely contained in the union of these envelope sets (note that it is not sufficient to merely guess those squares on the boundaries, and this makes the study much different from previous studies on similar problems). A challenge is: how to ensure that the guesses can be done in polynomial time. For the guessed squares, all those points contained in them should be counted. Another challenge is: how to guarantee that these points are not counted repeatedly.

The organization of the remaining parts of this paper is as follows. Section \ref{sec2} is the main part. In Subsection \ref{sec2.1}, we present the preliminaries of related problems and a relationship between Min$p$U and D$k$SH. In Subsection \ref{sec2.2}, we study a variant of the D$k$SH-US problem in a block, which serves as an auxiliary subproblem, and give a polynomial-time algorithm to compute an exact solution. In Subsection \ref{sec2.3}, local solutions to these subproblems are assembled to yield a $(4,\frac{1}{1+\varepsilon})$-bicriteria  approximate solution to the D$k$SH-US instance, which, by the previous relation, yields a $(\frac{1}{1+\varepsilon},4)$-bicriteria  approximate solution to the Min$p$U-US instance. Section \ref{sec3} concludes the paper with some discussions on future work.

\section{Approximation Algorithm for Min$p$U-US}\label{sec2}

\subsection{Preliminaries}\label{sec2.1}

In this subsection, we formally define the Min$p$U problem and the D$k$SH problem, and reveals a relation between then in terms of bicriteria algorithms.

\begin{definition}[Minimum $p$ Union (Min$p$U)]
{\rm Given a hypergraph $G=(V,E)$ and an integer $1\leq p\leq m$, where $m$ is the number of hyperedges in $E$, the goal of Min$p$U is to find a set of $p$ hyperedges $E'\subseteq E$ to cover the minimum number of vertices, that is, $|E'|=p$ such that $|\bigcup_{e\in E'}e|$ is minimized.}
\end{definition}

\begin{definition}[Densest $k$-Subhypergraph (D$k$SH) ]
{\rm Given a hypergraph $G=(V,E)$ and an integer $1\leq k\leq n$, where $n$ is the number of vertices in $V$, the goal of D$k$SH is to find a subset of vertices $V'\subseteq V$ of size at most $k$ that contains the largest number of hyperedges, in other words, $|E(V')|$ is maximized, where $E(V')=\{e\in E\colon e\subseteq V'\}$.}
\end{definition}

\begin{remark}\label{rem0718-1}
{\rm It should be remarked that we may assume that any feasible solution $V'$ to a D$k$SH instance satisfies
\begin{equation}\label{eq0708-1}
V'=\bigcup_{e\in E(V')}e.
\end{equation}
In fact, by the definition, we have $\bigcup_{e\in E(V')}e\subseteq V'$. Then vertex set $V''=\bigcup_{e\in E(V')}e$ satisfies $|V''|\leq |V'|\leq k$ and $E(V'')=E(V')$, and thus $V''$ is a feasible solution to the D$k$SH instance which is no worse than $V'$. So, it suffices to consider such $V''$.
In view of assumption \eqref{eq0708-1}, the D$k$SH problem is equivalent to finding the maximum number of hyperedges to cover at most $k$ vertices.}
\end{remark}

The next theorem presents a relationship between Min$p$U and D$k$SH.

\begin{theorem}\label{thm1226}
If there is a $(\beta,\alpha)$-bicriteria approximation algorithm for D$k$SH, then there is an $(\alpha,\beta)$-bicriteria approximation algorithm for Min$p$U.
\end{theorem}
\begin{proof}
In the following, we use APX and OPT to denote the approximate solution and an optimal solution, and use $apx$ and $opt$ to denote their objective values, respectively. A subscript is used to specify which problem the symbol is referring to.

Given a Min$p$U instance $(V,E,p)$, we construct an $(\alpha,\beta)$-bicriteria solution $E'$ as follows.
For each $k=1,\ldots,|V|$, call the $(\beta,\alpha)$-bicriteria approximation algorithm for D$k$SH to yield a vertex set $V_k$. Let $\ell$ be the smallest index satisfying
\begin{equation}\label{eq1226}
|E(V_{\ell})|\geq \alpha p.
\end{equation}
Then $E'=E(V_{\ell})$ is a set of hyperedges violating the feasibility of the Min$p$U instance by a factor of $\alpha$. Next, we show that $E'$ approximates the optimal value of the Min$p$U instance within factor $\beta$. By Remark \ref{rem0718-1}, this is equivalent to show that
\begin{equation}\label{eq0125-1}
|V_{\ell}|\leq \beta opt_{MinpU}.
\end{equation}

Since $V_{\ell}$ is computed by a $(\beta,\alpha)$-bicriteria algorithm for the D$\ell$SH instance, we have
\begin{equation}\label{eq0125}
|V_{\ell}|\leq \beta\ell.
\end{equation}
Consider the computed vertex set $V_k$ of the D$k$SH instance for $k=opt_{MinpU}$, we have
\begin{equation}\label{eq1226-2}
|E(V_k)|\geq \alpha opt_{DkSH}.
\end{equation}
Let $V_{MinpU}=\bigcup_{e\in OPT_{MinpU}}e$. Then $V_{MinpU}$ is a feasible solution to the D$k$SH instance containing at least $p$ hyperedges. Hence
\begin{equation}\label{eq1226-3}
opt_{DkSH}\geq p.
\end{equation}
Combining inequalities \eqref{eq1226-2} and \eqref{eq1226-3}, we have $|E(V_k)|\geq \alpha p$. By the choice of $\ell$, we have
\begin{equation}\label{eq1226-1}
\ell\leq opt_{MinpU}.
\end{equation}
Combining inequalities \eqref{eq0125} and \eqref{eq1226-1}, inequality \eqref{eq0125-1} is proved, nd the theorem follows.
\end{proof}

When $V$ is a set of points on the plane and $E$ corresponds to a set $\mathcal S$ of unit squares on the plane (that is, every $e\in E$ consists of all those points in a unit square $s_e\in\mathcal S$ corresponding to $e$), then the above problems are called {\em unit square Min$p$U} (Min$p$U-US) and {\em unit square D$k$SH} (D$k$SH-US), respectively.

By Theorem \ref{thm1226}, to design a $(\frac{1}{1+\varepsilon},4)$-bicriteria algorithm for Min$p$U-US, it suffices to design a $(4,\frac{1}{1+\varepsilon})$-bicriteria algorithm for D$k$SH-US. We employ the partition and shifting strategy: divide the area containing all the points into blocks of constant side-length, solve subproblems on the blocks, and then assemble the solutions to the subproblems into a feasible solution to the original problem. A crucial step is to design an algorithm for the subproblem on a block. It is done by a dynamic programming method.

\subsection{Algorithm for Subproblem in a Block}\label{sec2.2}

Let $b$ be a block of side-length $a\times a$, where $a$ is a constant, and $\mathcal S_b$ be the set of unit squares intersecting $b$. For simplicity of notation, we still write $\mathcal S$ for $\mathcal S_b$. For a subcollection $\mathcal S'\subseteq\mathcal S$, those points belonging to at least one unit square of $\mathcal S'$ are said to be {\em covered} by $\mathcal S'$. In this subsection, we consider the D$k$SH-US problem on block $b$ (denote the problem as D$k$SH-US$_b$). In view of Remark \ref{rem0718-1}, the problem can be stated as follows: for an integer $k_b$, find a set of unit squares $\mathcal S'\subseteq \mathcal S$ such that the number of points covered by $\mathcal S'$ is at most $k_b$ and subject to this constraint, $|\mathcal S'|$ is as large as possible.

Divide $b$ into grids of side-length 1 (for simplicity of statement, assume that the side-length $a$ is an integer). For any subcollection of unit squares $\mathcal S'\subseteq \mathcal S$, denote by $U(\mathcal S')=\bigcup_{s\in \mathcal S'}s$ the union region of $\mathcal S'$. For simplicity of notation, we also use $U(\mathcal S')$ to denote the set of points covered by $\mathcal S'$.

\begin{definition}[envelope]
{\rm For a set of unit squares $\mathcal S'$ and a grid point $g$, let $\mathcal S'_g$ be the set of unit squares in $\mathcal S'$ containing $g$. Those unit squares appearing on the boundary of $U(\mathcal S'_g)$ are called the {\em envelope-squares} of $\mathcal S'$ at $g$. The union region of those envelope-squares is called the {\em envelope of $\mathcal S'$ at $g$}.}
\end{definition}

It is assumed that the positions of the points and the positions of the unit squares are generic so that no point lies on the boundary of a unit square and no square have the same $x$-coordinate or $y$-coordinate.
As a consequence, we may assume that the unit squares are open, and thus {\em every unit square  belongs to exactly one grid point}.

For simplicity of statements, we use $s_g^b$ and $s_g^e$ to denote two {\em virtual squares} whose positions are to the left and to the right of all squares associated with grid point $g$, call them the {\em beginning square} and the {\em ending square} of $g$, respectively.

\subsubsection{The idea underlying the dynamic programming}

In this subsection, we use a series of examples to develop the ideas behind the dynamic programming, in the hope that the complicated symbols in the next subsection will not seem too abruptly.

To find out an optimal solution $\mathcal O^*$ to an D$k$SH-US$_b$ instance, it suffices to find out all unit squares in $\mathcal O^*_g$ for every grid point $g$, where $\mathcal O^*_g$ is the set of unit squares in $\mathcal O^*$ that contains grid point $g$. To find out $\mathcal O^*_g$, it suffices to find out the envelope of $\mathcal O^*$ at $g$, and all those unit squares completely contained in the envelope. The envelope can be discovered by moving a vertical line $\ell$ (called {\em sweep line}) from left to right, tracing the highest and the lowest squares it meet during the movement. Consider the instance in Fig. \ref{fig0407-1} $(a)$ for an illustration.
Denote by $s_h$ and $s_l$ the highest and the lowest unit squares met by $\ell$ when it is at some position. Moving $\ell$ rightward, the tuples $(s_h,s_l)$ met by $\ell$ are sequentially $(s_g^b,s_g^b),(s_2,s_2),(s_2,s_3),(s_2,s_4),(s_1,s_4),(s_1,s_1),(s_g^e,s_g^e)$. In order to find out all those unit squares completely contained in the envelope, we add an element $s_{next}$ to the tuple indicating the next square to be met by the sweep line.
For the above instance, the modified sequence of tuples $(s_h,s_l;s_{next})$, only considering the movement in the left side of grid point $g$, are $(s_g^b,s_g^b;s_2),(s_2,s_2;s_3),(s_2,s_3;s_5),(s_2,s_3;s_4),(s_2,s_4;s_1),(s_2,s_4;s_g^e)$.

\begin{figure}[h]
\begin{center}
\begin{picture}(100,100)
\put(50,0){\line(0,1){100}}\put(0,50){\line(1,0){100}}
\put(50,50){\circle*{3}}\put(51,54){$g$}
\put(40,35){\framebox(50,50)}\put(10,40){\framebox(50,50)}\put(17,20){\framebox(50,50)}\put(35,13){\framebox(50,50)}
\put(22,30){\framebox(50,50)}
\put(91,85){$s_1$}\put(3,93){$s_2$}\put(8,15){$s_3$}\put(85,8){$s_4$}\put(18,83){$s_5$}
\put(30,0){\line(0,1){102}}\put(31,95){\small $\ell$}
{\linethickness{0.3mm} \put(10,40){\line(0,1){50}}\put(10,90){\line(1,0){50}}
\put(60,90){\line(0,-1){5}}\put(60,85){\line(1,0){30}}\put(90,85){\line(0,-1){50}}
\put(90,35){\line(-1,0){5}}\put(85,35){\line(0,-1){22}}\put(85,13){\line(-1,0){50}}
\put(35,13){\line(0,1){7}}\put(35,20){\line(-1,0){18}}\put(17,20){\line(0,1){20}}\put(17,40){\line(-1,0){7}} }
\put(42,-8){(a)}
\end{picture}
\hskip 2cm
\begin{picture}(100,110)
\put(50,0){\line(0,1){100}}\put(0,50){\line(1,0){100}}
\put(50,50){\circle*{3}}\put(51,54){$g$}
\put(27,0){\line(0,1){102}}\put(28,95){\small $\ell^{\it left}$}
\put(77,0){\line(0,1){102}}\put(78,95){\small $\ell^{\it right}$}
\put(40,35){\framebox(50,50)}\put(10,40){\framebox(50,50)}\put(17,20){\framebox(50,50)}\put(35,13){\framebox(50,50)}
\put(22,30){\framebox(50,50)}
\put(91,85){$s_1$}\put(3,93){$s_2$}\put(8,15){$s_3$}\put(85,8){$s_4$}\put(16,83){$s_5$}
\put(42,-8){(b)}
\end{picture}
\vskip 0.2cm \caption{$(a)$ Illustration of using sweep line $\ell$ to trace squares. For the position of $\ell$, the highest square is $s_2$, the lowest square is $s_3$, and the next square is $s_4$. The blackened lines mark the envelope. $(b)$ Illustration of using two symmetric sweep lines $\ell^{left}$ and $\ell^{right}$.}\label{fig0407-1}
\end{center}
\end{figure}
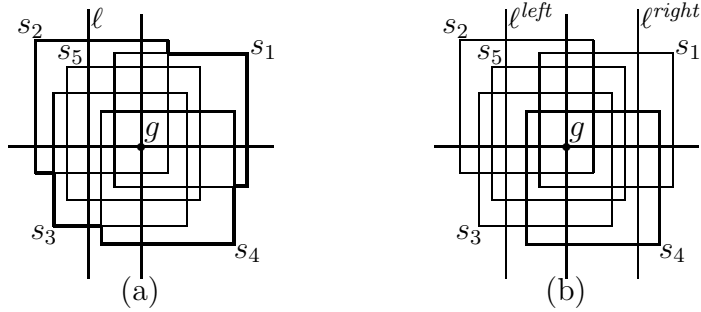

Accompanying the above tracing process, every time a square is recorded, ``all'' points contained in it should be counted. One crucial thing that should be paid attention to is: how to avoid repetition in the counting.
The idea is that when a new square is met, only those points ``newly'' covered by the new square are counted. That is, when moving from $(s_h^{old},s_l^{old};s_{next}^{old})$ to $(s_h^{new},s_l^{new};s_{next}^{new})$, we count those points in $(s_h^{new}\cup s_l^{new}\cup s_{next}^{new})\setminus (s_h^{old}\cup s_l^{old}\cup s_{next}^{old})$. However, such a technique cannot guarantee that every point is only counted once. Consider the instance in Fig. \ref{fig0421-1} for an example, in which newly covered points lie in the shaded areas. The tuples are sequentially $(s_g^b,s_g^b;s_1)$, $(s_1,s_1;s_2)$, $(s_2,s_1;s_3)$, $(s_3,s_1;s_4)$, $(s_4,s_1;s_5)$ and $(s_4,s_1;s_g^e)$ (again only consider the movement of the sweep line in the left-side of the grid point). During the movement, points counted are in the shaded areas $s_1$, $s_2\setminus s_1$, $s_3\setminus (s_1\cup s_2)$, $s_4\setminus (s_1\cup s_2\cup s_3)$, $s_5\setminus (s_1\cup s_3\cup s_4)$ and $s_g^e\setminus (s_1\cup s_4\cup s_5)=\emptyset$. Notice that points in the dark shaded area of $(b)$ and $(e)$ are counted twice. The reason for such a repetition is because when it is the time to count new points in the new square $s_5$, square $s_2$ is already ``forgotten'' by the previous tuple, and thus some points in $s_2\cap s_5$ are re-counted.

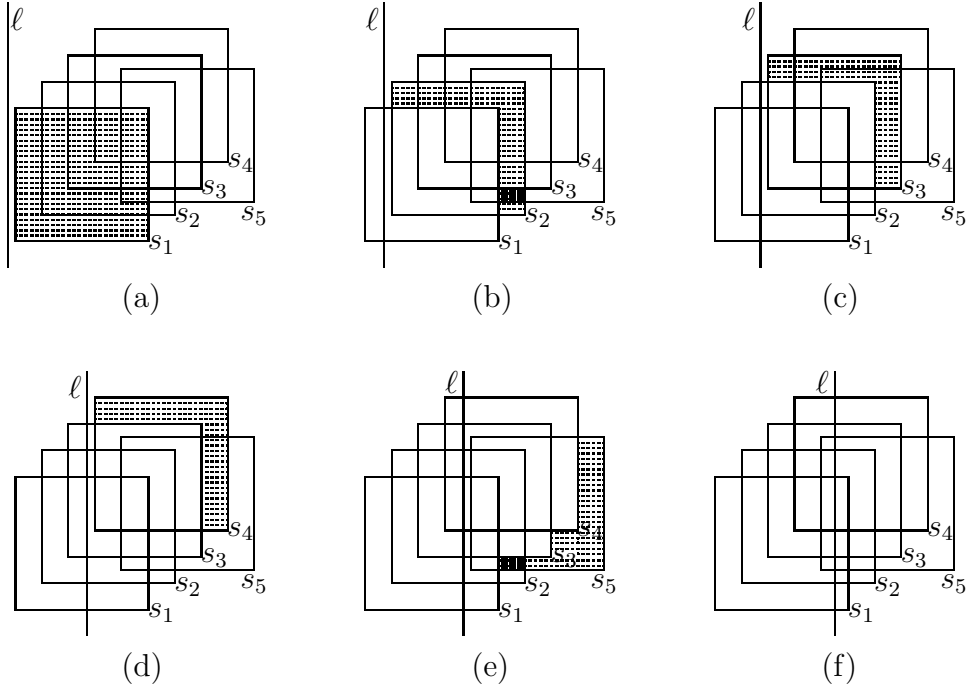
\begin{figure}[h]
\begin{center}
\hskip 0cm\begin{picture}(100,110)
\put(-3,10){\line(0,1){100}}\put(-2,100){$\ell$}
\put(0,20){\framebox(50,50)}\put(10,30){\framebox(50,50)}
\put(20,40){\framebox(50,50)}\put(30,50){\framebox(50,50)}\put(40,35){\framebox(50,50)}
\put(50,17){$s_1$}\put(60,27){$s_2$}\put(70,38){$s_3$}\put(80,48){$s_4$}\put(85,28){$s_5$}
\multiput(0,22)(0,4){12}{\dashbox(50,2)}
\put(40,-5){(a)}
\end{picture}
\hskip 1.cm\begin{picture}(100,110)
\put(7,10){\line(0,1){100}}\put(0,100){$\ell$}
\put(0,20){\framebox(50,50)}\put(10,30){\framebox(50,50)}
\put(20,40){\framebox(50,50)}\put(30,50){\framebox(50,50)}\put(40,35){\framebox(50,50)}
\put(50,17){$s_1$}\put(60,27){$s_2$}\put(70,38){$s_3$}\put(80,48){$s_4$}\put(85,28){$s_5$}
\multiput(10,76)(0,-4){2}{\dashbox(50,2)}
\multiput(50.2,68)(0,-4){10}{\dashbox(10,2)}
\multiput(50.5,35)(1,0){10}{\line(0,1){5}}
\put(40,-5){(b)}
\end{picture}
\hskip 1.cm\begin{picture}(100,110)
\put(17,10){\line(0,1){100}}\put(10,100){$\ell$}
\put(0,20){\framebox(50,50)}\put(10,30){\framebox(50,50)}
\put(20,40){\framebox(50,50)}\put(30,50){\framebox(50,50)}\put(40,35){\framebox(50,50)}
\put(50,17){$s_1$}\put(60,27){$s_2$}\put(70,38){$s_3$}\put(80,48){$s_4$}\put(85,28){$s_5$}
\multiput(20,86)(0,-4){2}{\dashbox(50,2)}
\multiput(60.2,78)(0,-4){10}{\dashbox(10,2)}
\put(40,-5){(c)}
\end{picture}
\vskip 1cm\hskip 0.cm\begin{picture}(100,110)
\put(27,10){\line(0,1){100}}\put(20,100){$\ell$}
\put(0,20){\framebox(50,50)}\put(10,30){\framebox(50,50)}
\put(20,40){\framebox(50,50)}\put(30,50){\framebox(50,50)}\put(40,35){\framebox(50,50)}
\put(50,17){$s_1$}\put(60,27){$s_2$}\put(70,38){$s_3$}\put(80,48){$s_4$}\put(85,28){$s_5$}
\multiput(30,96)(0,-4){2}{\dashbox(50,2)}
\multiput(70.2,88)(0,-4){10}{\dashbox(10,2)}
\put(40,-5){(d)}
\end{picture}
\hskip 1.cm\begin{picture}(100,110)
\put(37,10){\line(0,1){100}}\put(30,102){$\ell$}
\put(0,20){\framebox(50,50)}\put(10,30){\framebox(50,50)}
\put(20,40){\framebox(50,50)}\put(30,50){\framebox(50,50)}\put(40,35){\framebox(50,50)}
\put(50,17){$s_1$}\put(60,27){$s_2$}\put(70,38){$s_3$}\put(80,48){$s_4$}\put(85,28){$s_5$}
\multiput(50.2,36.5)(0,-4){1}{\dashbox(40,2)}
\multiput(70.2,47)(0,-4){2}{\dashbox(20,2)}
\multiput(80.2,81)(0,-4){11}{\dashbox(10,2)}
\multiput(50.5,35)(1,0){10}{\line(0,1){5}}
\put(40,-5){(e)}
\end{picture}
\hskip 1.cm\begin{picture}(100,110)
\put(45,10){\line(0,1){100}}\put(38,102){$\ell$}
\put(0,20){\framebox(50,50)}\put(10,30){\framebox(50,50)}
\put(20,40){\framebox(50,50)}\put(30,50){\framebox(50,50)}\put(40,35){\framebox(50,50)}
\put(50,17){$s_1$}\put(60,27){$s_2$}\put(70,38){$s_3$}\put(80,48){$s_4$}\put(85,28){$s_5$}
\put(40,-5){(f)}
\end{picture}
\vskip 0.2cm \caption{An illustration in which using one sweep line cannot avoid repetition of counting. }\label{fig0421-1}
\end{center}
\end{figure}

This example shows that the above strategy cannot avoid counting points repeatedly. To solve such a problem, we use two {\em symmetric} sweep lines which move synchronously to trace the squares, one for the left side of grid point $g$ and the other for the right side. To be more concrete, assume that the grid point $g$ has $x$-coordinate $x_g$. When the right sweep line $\ell^{right}$ locates at coordinate $x_g+x$, the location for the left sweep line $\ell^{left}$ is then at $x_g+x-1$ (see Fig. \ref{fig0407-1} $(b)$). Each tuple now contains five elements $(s_{hl},s_{ll},s_{hr},s_{lr};s_{next})$, where $(s_{hl},s_{ll})$ records the highest and the lowest squares met by $\ell^{left}$, $(s_{hr},s_{lr})$ records the highest and the lowest squares met by $\ell^{right}$, and $s_{next}$ is the next square to be met by $\ell^{left}$ (by the symmetric assumption on the sweep lines, the right boundary of $s_{next}$ is the next position to be met by $\ell^{right}$). Denote by $S(g)$ the 5-tuple at grid point $g$. Without ambiguity, we also use $S(g)$ to denote the set of squares in the 5-tuple at $g$. When the sweep lines move from $S^{old}(g)$ to $S^{new}(g)$, those points in $S^{new}(g)\setminus S^{old}(g)$ are counted.

For the instance in Fig. \ref{fig0421-1}, the new technique yields counting in Fig. \ref{fig0422-1}. The sequence of 5-tuples are $(s_g^b,s_g^b,s_4,s_1;s_1)$, $(s_1,s_1,s_4,s_2;s_2)$, $(s_2,s_1,s_4,s_5;s_3)$, $(s_3,s_1,s_4,s_5;s_4)$, $(s_4,s_1,$ $s_5,s_5;s_5)$ and $(s_4,s_1,s_g^e,s_g^e;s_g^e)$. The points counted during the process are in the shaded areas $s_1\cup s_4$, $s_2\setminus (s_1\cup s_4)$, $(s_3\cup s_5)\setminus (s_1\cup s_2\cup s_4)$ and $\emptyset$, $\emptyset$, $\emptyset$ (notice that for $(d)$, $S^{new}(g)\setminus S^{old}(g)=(s_1\cup s_3\cup s_4\cup s_5)\setminus (s_1\cup s_2\cup s_3\cup s_4\cup s_5)=\emptyset$, and similar argument for $(e)$ and $(f)$). Using such a recording method, every region in the envelope is counted exactly once. In particular, when $s_5$ comes into the sight, square $s_2$ is remembered by the previous tuple, this is why repetition can be avoided.

\vskip 0.2cm \begin{figure}[h]
\begin{center}
\begin{picture}(100,110)
\put(-3,5){\line(0,1){105}}\put(-2,105){\small $\ell^{\it left}$}
\put(47,5){\line(0,1){105}}\put(48,105){\small $\ell^{\it right}$}
\put(0,20){\framebox(50,50)}\put(10,30){\framebox(50,50)}
\put(20,40){\framebox(50,50)}\put(30,50){\framebox(50,50)}\put(40,35){\framebox(50,50)}
\put(50,17){$s_1$}\put(60,27){$s_2$}\put(70,38){$s_3$}\put(80,48){$s_4$}\put(85,28){$s_5$}
\multiput(0,22)(0,4){12}{\dashbox(50,2)}
\multiput(30,52)(0,4){12}{\dashbox(50,2)}
\put(40,-5){(a)}
\end{picture}
\hskip 1.cm\begin{picture}(100,110)
\put(7,5){\line(0,1){105}}\put(8,105){\small $\ell^{\it left}$}
\put(57,5){\line(0,1){105}}\put(58,105){\small $\ell^{\it right}$}
\put(0,20){\framebox(50,50)}\put(10,30){\framebox(50,50)}
\put(20,40){\framebox(50,50)}\put(30,50){\framebox(50,50)}\put(40,35){\framebox(50,50)}
\put(47,13){$s_1$}\put(60,27){$s_2$}\put(70,38){$s_3$}\put(80,48){$s_4$}\put(85,28){$s_5$}
\multiput(10,76)(0,-4){2}{\dashbox(20,2)}
\multiput(50,47)(0,-4){5}{\dashbox(10,2)}
\put(40,-5){(b)}
\end{picture}
\hskip 1.cm\begin{picture}(100,110)
\put(17,5){\line(0,1){105}}\put(18,105){\small $\ell^{\it left}$}
\put(67,5){\line(0,1){105}}\put(68,105){\small $\ell^{\it right}$}
\put(0,20){\framebox(50,50)}\put(10,30){\framebox(50,50)}
\put(20,40){\framebox(50,50)}\put(30,50){\framebox(50,50)}\put(40,35){\framebox(50,50)}
\put(50,17){$s_1$}\put(57,24){$s_2$}\put(70,38){$s_3$}\put(80,48){$s_4$}\put(85,28){$s_5$}
\multiput(20,86)(0,-4){2}{\dashbox(10,2)}
\multiput(80.2,81)(0,-4){9}{\dashbox(10,2)}
\multiput(60,45)(0,-4){3}{\dashbox(30,2.2)}
\put(40,-5){(c)}
\end{picture}
\vskip 1cm\hskip 0.cm\begin{picture}(100,110)
\put(27,5){\line(0,1){105}}\put(28,105){\small $\ell^{\it left}$}
\put(77,5){\line(0,1){105}}\put(68,105){\small $\ell^{\it right}$}
\put(0,20){\framebox(50,50)}\put(10,30){\framebox(50,50)}
\put(20,40){\framebox(50,50)}\put(30,50){\framebox(50,50)}\put(40,35){\framebox(50,50)}
\put(50,17){$s_1$}\put(60,27){$s_2$}\put(70,38){$s_3$}\put(80,48){$s_4$}\put(85,28){$s_5$}
\put(40,-5){(d)}
\end{picture}
\hskip 1.cm\begin{picture}(100,110)
\put(37,5){\line(0,1){105}}\put(38,105){\small $\ell^{\it left}$}
\put(87,5){\line(0,1){105}}\put(88,105){\small $\ell^{\it right}$}
\put(0,20){\framebox(50,50)}\put(10,30){\framebox(50,50)}
\put(20,40){\framebox(50,50)}\put(30,50){\framebox(50,50)}\put(40,35){\framebox(50,50)}
\put(50,17){$s_1$}\put(60,27){$s_2$}\put(70,38){$s_3$}\put(77,44){$s_4$}\put(88,28){$s_5$}
\put(40,-5){(e)}
\end{picture}
\hskip 1.cm\begin{picture}(100,110)
\put(47,5){\line(0,1){105}}\put(48,105){\small $\ell^{\it left}$}
\put(97,5){\line(0,1){105}}\put(98,105){\small $\ell^{\it right}$}
\put(0,20){\framebox(50,50)}\put(10,30){\framebox(50,50)}
\put(20,40){\framebox(50,50)}\put(30,50){\framebox(50,50)}\put(40,35){\framebox(50,50)}
\put(50,17){$s_1$}\put(60,27){$s_2$}\put(70,38){$s_3$}\put(80,48){$s_4$}\put(85,28){$s_5$}
\put(40,-5){(f)}
\end{picture}
\vskip 0.2cm \caption{Illustration for counting points using two sweep lines.}\label{fig0422-1}
\end{center}
\end{figure}
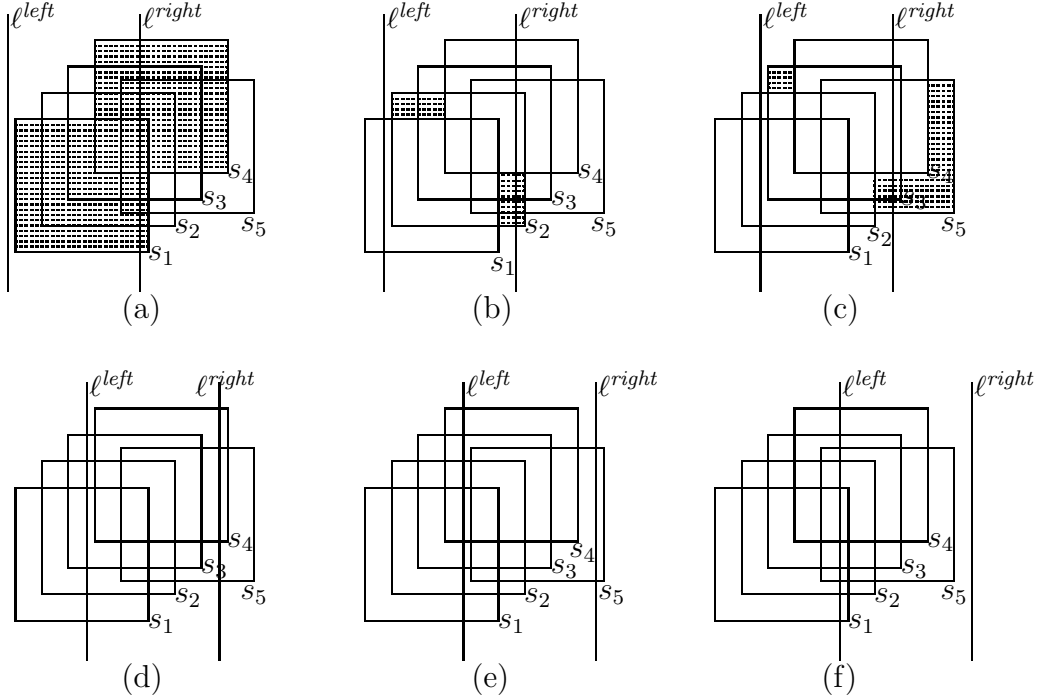

The above ideas are only illustrated by considering ``one'' grid point. One may be wondering what if there are interactions among several grid points? The idea is, for each grid point $g$, two sweep lines $\ell^{left}_g,\ell^{right}_g$ are used, and all sweep lines $\{(\ell^{left}_g,\ell^{right}_g)\}_g$ move synchronously rightwards. For a position of $\{(\ell^{left}_g,\ell^{right}_g)\}_g$, denote $S=\bigcup_gS(g)$. When the sweep lines move from $S^{old}$ to $S^{new}$, those points in $S^{new}\setminus S^{old}$ are counted. An illustration is given in Fig. \ref{fig0423-1} (we only draw out the figures for the first three positions of the sweep lines), and the 5-tuples are listed in Table \ref{tab0422-1}. A star in the table indicates that this tuple is different from its predecessor (note that for every movement of sweep lines, only one tuple is different from its predecessor). During the movement, the counted areas are sequentially $s_1\cup s_3\cup s_4\cup s_6 \cup s_7$, $s_2\setminus (s_1\cup s_3\cup s_4\cup s_6 \cup s_7)$, $s_5\setminus (s_1\cup s_2\cup s_3\cup s_4\cup s_6 \cup s_7)$ and $\emptyset,\emptyset,\emptyset,\emptyset,\emptyset$. Notice that when the sweep lines move from the positions in $(a)$ to the positions in $(b)$, if we are only considering grid point $g_1$, then both the shaded area and the two boxes bounded by the darkened lines are newly covered. But considering the interactions of those squares associated with $g_2$ and $g_3$, only the shaded area is newly covered. The same argument for the figure in $(c)$.

\begin{figure}[htbp]
\begin{center}
\hskip -0.6cm\begin{picture}(140,180)
\multiput(70,6)(50,0){2}{\multiput(0,0)(0,5){35}{\line(0,1){3}}}
\multiput(20,55)(0,50){2}{\multiput(0,0)(5,0){27}{\line(1,0){3}}}
\put(60,48){$g_1$}\put(60,109){$g_2$}\put(110,48){$g_3$}\put(110,109){$g_4$}
\put(25,9){\line(0,1){165}}\put(26,165){\small $\ell_{g_1,g_2}^{\it left}$}
\put(75,9){\line(0,1){165}}\put(76,170){\small $\ell_{g_1,g_2}^{\it right}$}\put(76,155){\small $=\ell_{g_3,g_4}^{\it left}$}
\put(125,9){\line(0,1){165}}\put(126,165){\small $\ell_{g_3,g_4}^{\it right}$}
\put(30,26){\framebox(50,50)}\put(42,36){\framebox(50,50)}\put(54,46){\framebox(50,50)}\put(90,12){\framebox(50,50)}
\put(36,66){\framebox(50,50)}\put(48,82){\framebox(50,50)}\put(60,100){\framebox(50,50)}
\put(76,19){$s_1$}\put(88,29){$s_2$}\put(100,39){$s_3$}\put(132,6){$s_7$}
\put(30,119){$s_4$}\put(42,135){$s_5$}\put(55,153){$s_6$}
\multiput(30,28)(0,4){12}{\dashbox(50,2)}\multiput(54,48)(0,4){12}{\dashbox(50,2)}\multiput(90,14)(0,4){12}{\dashbox(50,2)}
\multiput(36,68)(0,4){12}{\dashbox(50,2)}\multiput(60,102)(0,4){12}{\dashbox(50,2)}
\put(62,-5){(a)}
\end{picture}
\hskip 0.3cm\begin{picture}(140,180)
\multiput(70,6)(50,0){2}{\multiput(0,0)(0,5){35}{\line(0,1){3}}}
\multiput(20,55)(0,50){2}{\multiput(0,0)(5,0){27}{\line(1,0){3}}}
\put(60,48){$g_1$}\put(60,109){$g_2$}\put(110,48){$g_3$}\put(110,109){$g_4$}
\put(33,9){\line(0,1){165}}\put(34,165){\small $\ell_{g_1,g_2}^{\it left}$}
\put(83,9){\line(0,1){165}}\put(84,170){\small $\ell_{g_1,g_2}^{\it right}$}\put(84,155){\small $=\ell_{g_3,g_4}^{\it left}$}
\put(133,9){\line(0,1){165}}\put(134,165){\small $\ell_{g_3,g_4}^{\it right}$}
\put(30,26){\framebox(50,50)}\put(42,36){\framebox(50,50)}\put(54,46){\framebox(50,50)}\put(90,12){\framebox(50,50)}
\put(36,66){\framebox(50,50)}\put(48,82){\framebox(50,50)}\put(60,100){\framebox(50,50)}
\put(72,19){$s_1$}\put(90,29){$s_2$}\put(100,39){$s_3$}\put(132,6){$s_7$}
\put(33,119){$s_4$}\put(42,135){$s_5$}\put(49,148){$s_6$}
\multiput(80.2,37.4)(0,2.9){3}{\dashbox(9.6,1.4)}\put(90,36.6){\linethickness{0.3mm}\framebox(2,8.6)}
\put(42.6,76.8){\linethickness{0.3mm}\framebox(10.5,8.6)}
\put(62,-5){(b)}
\end{picture}
\hskip 0.3cm\begin{picture}(140,180)
\multiput(70,6)(50,0){2}{\multiput(0,0)(0,5){35}{\line(0,1){3}}}
\multiput(20,55)(0,50){2}{\multiput(0,0)(5,0){27}{\line(1,0){3}}}
\put(60,48){$g_1$}\put(60,109){$g_2$}\put(110,48){$g_3$}\put(110,109){$g_4$}
\put(38,9){\line(0,1){165}}\put(40,165){\small $\ell_{g_1,g_2}^{\it left}$}
\put(88,9){\line(0,1){165}}\put(90,170){\small $\ell_{g_1,g_2}^{\it right}$}\put(90,155){\small $=\ell_{g_3,g_4}^{\it left}$}
\put(138,9){\line(0,1){165}}\put(140,165){\small $\ell_{g_3,g_4}^{\it right}$}
\put(30,26){\framebox(50,50)}\put(42,36){\framebox(50,50)}\put(54,46){\framebox(50,50)}\put(90,12){\framebox(50,50)}
\put(36,66){\framebox(50,50)}\put(48,82){\framebox(50,50)}\put(60,100){\framebox(50,50)}
\put(76,19){$s_1$}\put(90,29){$s_2$}\put(100,39){$s_3$}\put(138,6){$s_7$}
\put(27,119){$s_4$}\put(42,135){$s_5$}\put(49,148){$s_6$}
\multiput(48,118)(0,3.6){4}{\dashbox(11.8,1.8)}
\multiput(86.2,97.5)(0,4){1}{\dashbox(12,1.2)}
\put(86.8,82.2){\linethickness{0.3mm}\framebox(11.,13.5)}
\put(62,-5){(c)}
\end{picture}

\vskip 0.2cm \caption{Illustration of counting with interactions among several grid points. Shaded areas indicate newly covered regions.}\label{fig0423-1}
\end{center}
\end{figure}
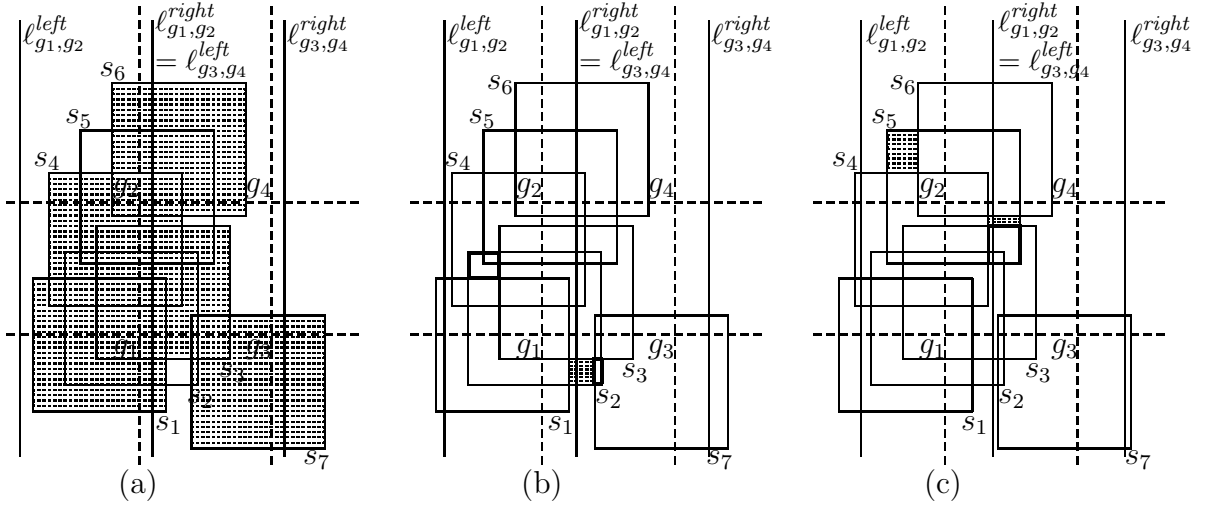

\begin{table}[h]
\centering
\begin{tabular}{|c|c|c|c|}
\hline
 & $g_1$ & $g_2$ & $g_3$ \\
\hline
$(a)$ & $(s_{g_1}^b,s_{g_1}^b,s_3,s_1;s_1)$ & $(s_{g_2}^b,s_{g_2}^b,s_6,s_4;s_4)$ & $(s_{g_3}^b,s_{g_3}^b,s_7,s_7;s_7)$ \\
$(b)$ & * $(s_1,s_1,s_3,s_2;s_2)$ & $(s_{g_2}^b,s_{g_2}^b,s_6,s_4;s_4)$ & $(s_{g_3}^b,s_{g_3}^b,s_7,s_7;s_7)$ \\
$(c)$ & $(s_1,s_1,s_3,s_2;s_2)$ & * $(s_4,s_4,s_6,s_5;s_5)$ & $(s_{g_3}^b,s_{g_3}^b,s_7,s_7;s_7)$ \\
& $(s_1,s_1,s_3,s_2;s_2)$ & $(s_4,s_4,s_6,s_5;s_5)$ & * $(s_7,s_7,s_{g_3}^e,s_{g_3}^e;s_{g_3}^e)$ \\
& * $(s_2,s_1,s_3,s_3;s_3)$ & $(s_4,s_4,s_6,s_5;s_5)$ & $(s_7,s_7,s_{g_3}^e,s_{g_3}^e;s_{g_3}^e)$ \\
& $(s_2,s_1,s_3,s_3;s_3)$ & * $(s_5,s_4,s_6,s_6;s_6)$ & $(s_7,s_7,s_{g_3}^e,s_{g_3}^e;s_{g_3}^e)$ \\
& * $(s_3,s_1,s_{g_1}^e,s_{g_1}^e;s_{g_1}^e)$ & $(s_5,s_4,s_6,s_6;s_6)$ & $(s_7,s_7,s_{g_3}^e,s_{g_3}^e;s_{g_3}^e)$ \\
& $(s_3,s_1,s_{g_1}^e,s_{g_1}^e;s_{g_1}^e)$ & * $(s_6,s_4,s_{g_2}^e,s_{g_2}^e;s_{g_2}^e)$ & $(s_7,s_7,s_{g_3}^e,s_{g_3}^e;s_{g_3}^e)$\\
\hline
\end{tabular}
  \caption{The 5-tuples for the example of Fig. \ref{fig0423-1}.}\label{tab0422-1}
\end{table}

These ideas are formally described in the following subsection.

\subsubsection{The Dynamic Programming}\label{sec0708-1}

Let $\ell(x)$ be the sweep line with the horizontal coordinate $x$, and $x(\ell)$ be the $x$-coordinate of sweep line $\ell$. For a unit square $s$, use $x(s)$ to denote the $x$-coordinate of the right boundary and $y(s)$ to denote the $y$-coordinate of the upper boundary of $s$.

\begin{definition}[configuration]\label{def0131}
{\rm A {\em configuration} in block $b$ consists of a 3-tuple $(S,\tilde k, x)$, where $S=\bigcup_{g}S(g)$, $S(g)=(s_{hl}(g),s_{ll}(g),s_{hr}(g),s_{lr}(g);s_{next}(g))$ is a 5-tuple of unit squares associated with grid point $g$, $0\leq \tilde k\leq k_b$ is an integer, and $x$ is a real number in $[0,1]$. For each grid point $g$, let $\ell_g^{right}$ and $\ell_g^{left}$ be two vertical lines with $x(\ell_g^{right})=x(g)+x$ and $x(\ell_g^{left})=x(\ell_g^{right})-1$. The 5-tuple at $g$ satisfies $y(s_{hl}(g))\geq y(s_{ll}(g))$, $y(s_{hr}(g))\geq y(s_{lr}(g))$, $y(s_{lr}(g))\leq y(s_{next}(g))\leq y(s_{hr}(g))$, $\max\{x(s_{hl}(g)),x(s_{ll}(g))\}< x(\ell_g^{right}) < \min\{x(s_{hr}(g)),x(s_{lr}(g)),x(s_{next}(g))\}$. Furthermore, there exists a grid point $g$ such that $x(\ell_g^{right})=x(s_{next}(g))-\varepsilon$, where $\varepsilon$ is a sufficiently small constant.}
\end{definition}

\begin{remark}\label{rem0719-3}
{\rm The conditions in the definition of configuration are satisfied by any set of guessed squares. For example, The reason for $\max\{x(s_{hl}(g)),x(s_{ll}(g))\}< x(\ell_g^{right})$ is because: $s_{hl}(g)$ and $s_{ll}(g)$ are cut by $\ell^{left}_g$, and thus $\ell^{right}_g$ is behind the right boundary of $s_{hl}(g)$ and $s_{ll}(g)$. Similarly, $s_{hr}(g),s_{lr}(g)$ are cut by $\ell^{right}_g$, and thus $x(\ell^{right}_g)<\min\{x(s_{hr}(g)),x(s_{lr}(g))\}$. Because the left boundary of $s_{next}(g)$ is not reached by $\ell^{left}_g$, line $\ell^{right}_g$ must cut through $s_{next}(g)$, and thus $x(\ell^{right}_g)<x(s_{next}(g))$. The reason for $y(s_{lr}(g))\leq y(s_{next}(g))\leq y(s_{hr}(g))$ is because: $s_{hr}(g)$ and $s_{lr}(g)$ are the highest and the lowest squares cut by $\ell^{right}_g$, and $s_{next}(g)$ is also cut by $\ell^{right}_g$, so $s_{next}(g)$ cannot be higher than $s_{hr}(g)$ or lower than $s_{lr}(g)$.}
\end{remark}

\begin{remark}
{\rm Note that although $x$ in the definition of configuration is a real number, the last requirement on the position of $\ell_g^{right}$ shows that configurations can be discretized by only considering right boundaries of those unit squares. The reason why we use a small constant $\varepsilon$ is to let $s_{next}(g)$ to be a ``next'' square to be met. The role of integer $\tilde k$ is to guess the number of covered points, which will be clear after the definition of an auxiliary directed acyclic graph (DAG) in the following.}
\end{remark}

The DAG with parameter $k_b$ restricted to a block $b$ is constructed as follows and some explanations are given after the construction. For a configuration $u$, we use $\ell^{left}_{u,g}$ and $\ell^{right}_{u,g}$ to denote the sweep lines associated with grid point $g$ in configuration $u$.

\begin{definition}[DAG]\label{def0131-1}
{\rm The vertex set of the auxiliary digraph $G_{b,k_b}$ in block $b$ with parameter $k_b$ consists of all configurations, a source vertex $u_{src}$, and a sink vertex $u_{sink}$. The arcs in $G_{b,k_b}$ are as follows.

({\em Arc between two configurations}) For two configurations $u=(S^{u},\tilde k^{u},x_u)$ and $v=(S^{v},\tilde k^{v},x_v)$, there is an arc $(u,v)$ in $G_{b,k_b}$ if only if all the following conditions hold:

($\romannumeral1$) Let $g'=\arg\min_g\{x(s_{next}^{u}(g)) -x(g)\}$ and $g''=\arg\min_{g:g\neq g'}\{x(s_{next}^{u}(g)) -x(g)\}$. It is required that $x(s_{next}^{u}(g')) -x(g')-\varepsilon<x_v\leq x(s_{next}^{u}(g'')) -x(g'')-\varepsilon$. This condition reflects the requirement that sweep lines must move {\em step by step}, an arc $(u,v)$ is possible only when a sweep line strides over the left boundary of the first $s_{next}^u$, and the movement cannot be too large to stride over the left boundary of the second $s_{next}^u$.

($\romannumeral2$) The token sets in $S^u$ and $S^v$ satisfy the following conditions:

($\romannumeral2_1$) $S^u$ and $S^v$ differ in exactly one grid point, namely $g'$ defined in $(\romannumeral1)$. So, in the following conditions, if a 5-tuple is changed, it always refers to squares associated with $g'$.

($\romannumeral2_2$) ({\em monotonicity of $x$-coordinate}) For any grid point $g$ and any subscript $c\in\{hl,ll,hr,lr,next\}$, $x(s^{v}_{c}(g))\geq x(s^{u}_{c}(g))$. For any subscript $c\in\{lr,hr\}$, if $x(\ell^{right}_{v,g})<x(s_c^u(g))$, then $s_c^u(g)=s_c^v(g)$. For the grid point $g'$ in ($\romannumeral1$), $x(s^{v}_{next}(g'))> x(s^{u}_{next}(g'))$.

($\romannumeral2_3$) ({\em monotonicity of $y$-coordinate}) For any grid point $g$, $y(s_{hl}^{u}(g))\leq y(s_{hl}^{v}(g))$, $y(s_{ll}^{u}(g))\geq y(s_{ll}^{v}(g))$, $y(s_{hr}^{u}(g))\geq y(s_{hr}^{v}(g))$ and $y(s_{lr}^{u}(g))\leq y(s_{lr}^{v}(g))$. For any subscript $c\in\{hr,lr\}$, if $x(\ell^{right}_{v,g})>x(s_c^u(g))$, then $y(s_{hl}^v(g))\geq y(s_c^u(g))\geq y(s_{ll}^v(g))$. For grid point $g'$ in $(\romannumeral1)$, $y(s_{hl}^v(g'))\geq y(s_{next}^u(g'))\geq y(s_{ll}^v(g'))$.

($\romannumeral3$) $\tilde k^{v}$ is the sum of $\tilde k^{u}$ and the number of points which are covered by $U(S^{v})\setminus U(S^{u})$.

The weight on arc $(u,v)$ is set to be $w(u,v)=|S^v\setminus S^u|$.

({\em Arcs from source vertex to configurations}) Let $\ell_{0}^{left}(g)$ and $\ell_{0}^{right}(g)$ be the sweep lines with $x(\ell_{0}^{right}(g))=x_g$ and $x(\ell_{0}^{left}(g))=x_g-1$ (they are the beginning positions of sweep lines). The source vertex $s_{src}$ is linked to every vertex $u=(S^{u},\tilde k^{u},0)$ (beginning configurations corresponding to beginning positions), where every $S^{u}$ has the form $(s_{g}^b,s_{g}^b,s_{hr}(g),s_{lr}(g); s_{next}(g))$ (with virtual square $s_g^b$ serving as $s_{hl}(g)$ and $s_{ll}(g)$), and $\tilde k^{u}$ is the number of points covered by $S^{u}$. The weight on such an arc is $w(s_{src},u)=|S^u|$.

({\em Arcs from configurations to sink vertex}) A vertex $v=(S^v,\tilde k^v,1)$ is linked to the sink vertex $s_{sink}$ only when $S^v$ has the form $(s_{hl}(g),s_{ll}(g),s_{g}^e, s_{g}^e;s_{g}^e)$ (with virtual square $s_g^e$ serving as $s_{hr}(g)$, $s_{lr}(g)$ and $s_{next}(g)$), and $\tilde k^v=k_b$. The weight on such an arc is set to be $w(v,s_{sink})=0$.}
\end{definition}

We shall show later a relation between a source-sink path in the DAG and a solution to the D$k$SH-US$_b$ instance. As an illustration, the following is a source-sink path for the DAG of the example in Fig. \ref{fig0229-1} with $k_b=8$: $u_{scr}u_0u_1u_2u_3u_4u_5u_6u_7u_{sink}$, whose tokens are indicated by Table \ref{tab1}. For this path, the arcs have weights $w(u_{src},u_0)={|\{s_1,s_6\}|}=2$, $w(u_0,u_1)={|\{s_7,s_2\}|}=2$, $w(u_1,u_2)={|\{s_3\}|}=1$, $w(u_2,u_3)={|\{s_4\}|}=1$, $w(u_3,u_4)={|\{s_5\}|}=1$, and $w(u_4,u_{5})=w(u_5,u_{6})=w(u_6,u_{7})=w(u_7,u_{sink})=0$.

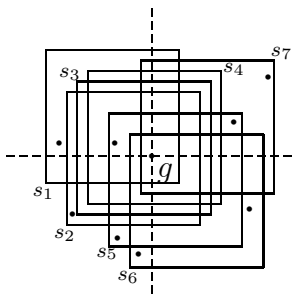
\begin{figure}[h]
\begin{center}
\begin{picture}(110,110)
\put(15,45){\framebox(50,50)}
\put(23,29){\framebox(50,50)}
\put(27,33){\framebox(50,50)}
\put(31,37){\framebox(50,50)}
\put(39,21){\framebox(50,50)}
\put(47,13){\framebox(50,50)}
\put(51,41){\framebox(50,50)}\put(57,47){$g$}\put(55,55){\circle*{2}}
{\scriptsize\put(10,40){$s_1$}\put(18,24){$s_2$}\put(20,85){$s_3$}
\put(82,87){$s_4$}\put(34,17){$s_5$}\put(42,8){$s_6$}\put(100,93){$s_7$} }

\multiput(55,3)(0,5){22}{\line(0,1){3}}
\multiput(0,55)(5,0){22}{\line(1,0){3}}
\put(20,60){\circle*{2}}\put(25,33){\circle*{2}}\put(92,35){\circle*{2}}
\put(42,24){\circle*{2}}\put(86,68){\circle*{2}}\put(99,85){\circle*{2}}
\put(50,18){\circle*{2}}\put(41,60){\circle*{2}}
\end{picture}
\vskip 0.5cm \caption{An illustration for the transition of configurations.  }\label{fig0229-1}
\end{center}
\end{figure}

\begin{table}[h]
\centering
\begin{tabular}{|c|c|c|c|}
\hline
 & {$x^u$} & {$\tilde k^u$} & {$S^u$} \\
\hline
$u_0$ & $0$ & 4 & $(s^b,s^b,s_1,s_6;s_1)$\\
\hline
$u_1$ & $x(s_2)-x(g)-\varepsilon$ & 7 & $(s_1,s_1,s_7,s_6;s_2)$ \\
\hline
$u_2$ & $x(s_3)-x(g)-\varepsilon$ & 7 & $(s_1,s_2,s_7,s_6;s_3)$\\
\hline
$u_3$ & $x(s_4)-x(g)-\varepsilon$ & 7 & $(s_1,s_2,s_7,s_6;s_4)$\\
\hline
$u_4$ & $x(s_5)-x(g)-\varepsilon$ & 8 & $(s_1,s_2,s_7,s_6;s_5)$\\
\hline
$u_5$ & $x(s_6)-x(g)-\varepsilon$ & 8 & $(s_1,s_5,s_7,s_6;s_6)$ \\
\hline
$u_6$ & $x(s_7)-x(g)-\varepsilon$ & 8 & $(s_1,s_6,s_7,s_7;s_7)$ \\
\hline
$u_7$ & 1 & 8 & $(s_1,s_6,s^e,s^e;s^e)$ \\
\hline
\end{tabular}
  \caption{A source-sink path in the DAG for the example in Fig. \ref{fig0229-1}.}\label{tab1}
\end{table}

\begin{remark}\label{rem0719-6}
{\rm It should be noted that all those conditions defining arcs of the DAG are satisfied by the transition of an optimal solution. For example, the first sentence of condition $(\romannumeral2_2)$ reflects the movement of sweep lines from left to right. The second sentence of condition $(\romannumeral2_2)$ holds for an optimal guessing at grid point $g$, because  $x(\ell^{right}_{v,g})<x(s_c^u(g))$ with $c\in\{hr,lr\}$ implies that $\ell^{right}_{v,g}$ {\em cuts through} $s_c^u(g)$ and by the monotonicity of the envelope, only when $\ell^{right}_g$ strides over the right boundary of $s_c^u(g)$, can $s_c(g)$ alter to another square. The third sentence of condition $(\romannumeral2_2)$ reflects the fact that a 5-tuple is altered only when the sweep line strides over the boundary of $s_{next}(g')$. The first sentence of $(\romannumeral2_3)$ reflects the monotonicity of $y$-coordinate of the envelope: in the left side of the grid point, $s_{hl}$ is higher and higher, and $s_{ll}$ is lower and lower; while in the right side of the grid point, $s_{hr}$ is lower and lower, $s_{lr}$ is higher and higher. The second sentence of condition $(\romannumeral2_3)$ holds because: $x(\ell^{right}_{v,g})>x(s_c^u(g))$ implies that $\ell^{left}_{v,g}$ cuts through $s_c^u(g)$, and thus $s_c^u(g)$ is no higher than $s_{hl}^v(g)$ and no lower than $s_{ll}^v(g)$.} The third sentence of $(\romannumeral2_3)$ holds because $g'$ is the grid point for which the 5-tuple is altered, which implies that $\ell^{left}_{g'}$ has stridden over the left boundary of $s_{next}^u(g')$, and thus $s_{next}^u(g')$ is cut by $\ell^{left}_{v,g'}$.
\end{remark}

For any source-sink path $P=u_{src}u_0u_1\ldots u_tu_{sink}$ in $G_{b,k_b}$, denote by $\mathcal S(P)=\bigcup_{i=0}^tS^{u_i}$. The following lemma shows that any source-sink path $P$ in $G_{b,k_b}$ corresponds to a feasible solution to the D$k$SH-US$_b$ instance whose weight equals $|\mathcal S(P)|$.

\begin{lemma}\label{lem0217}
Let $P=u_{src}u_0u_1\ldots u_tu_{sink}$ be a source-sink path in $G_{b,k_b}$. Then the set of unit squares $\mathcal S(P)=\bigcup_{i=0}^tU(S^{u_i})$ covers at most $k_b$ points. Furthermore,
$$w(P)=|\mathcal S(P)|.$$
\end{lemma}
\begin{proof}
By the definition of DAG in Definition \ref{def0131-1}, especially ($\romannumeral3$) and the existence of arc $(u_t, u_{sink})$, the first half of the lemma follows from the observation that every point covered by $\mathcal S(P)$ is counted into the parameter $k_b$.

By the method of assigning weight, we have $w(u_{i-1},u_{i})=|S^{u_i}\setminus S^{u_{i-1}}|$ for $i=0,\ldots,t$ and $w(u_t,u_{sink})=0$, where $u_{-1}=u_{src}$ and $S^{u_{-1}}=\emptyset$. So $w(P)=\sum_{i=0}^t|S^{u_i}\setminus S^{u_{i-1}}|\geq |\mathcal S(P)|$, the left hand side inequality is proved.

To prove the right hand side inequality, we show that each square $s\in \mathcal S(P)$ is counted exactly once in $w(P)$. For this purpose, assume that the first configuration that contains $s$ is $S^{u_i}$, and $s$ keeps to be in the token sets of a {\em consecutive} segment of the path until it disappears from some configuration, say $S^{u_j}$. Note that $s$ is counted as a new square in $S^{u_i}$, and is not counted anymore between $S^{u_i}$ and $S^{u_j}$. In the following, we prove that
\begin{equation}\label{eq0719-1}
\mbox{for any $j'\geq j$, $s$ will not appear in $S^{u_{j'}}$.}
\end{equation}
Suppose $s$ is associated with grid point $g$. We prove \eqref{eq0719-1} by distinguishing three cases.

{\bf Case 1:} $s$ appears in $S^{u_i}$ as $s_{hl}$ or $s_{ll}$.

Suppose $s=s_{hl}^{u_i}(g)$ (the argument for the case when $s=s_{ll}^{u_i}(g)$ is similar). By the monotonicity of $x$-coordinate (see condition ($\romannumeral2_2$) in the construction of DAG), after $s$ is replaced by another square $s'$ to serve as $s_{hl}$, we have $x(s)<x(s')$. In particular,
\begin{equation}\label{eq0509-1}
x(s)<x(s_{hl}^{u_j}(g))\leq x(s_{hl}^{u_{j'}}(g)),
\end{equation}
and thus $s$ will not serve as $s_{hl}^{u_{j'}}(g)$. By the definition of configuration, we have $x(s_{hl}^{u_{j'}}(g))<$ $\min\{x(s_{next}^{u_{j'}}(g)),x(s_{hr}^{u_{j'}}(g)),x(s_{lr}^{u_{j'}}(g))\}$. Combining this with \eqref{eq0509-1}, $s$ will not serve as $s_{next}^{u_{j'}}(g)$, $s_{hr}^{u_{j'}}(g)$, $s_{lr}^{u_{j'}}(g)$. If $s_{ll}^{u_i}(g)=s$, similar to the above, $s$ will not serve as $s_{ll}$ in any configuration after $u_j$. Next consider the case when $s_{ll}^{u_i}(g)\neq s$. By the monotonicity of $y$-coordinate (see $(\romannumeral2_3)$ in the construction of DAG),
\begin{equation}\label{eq0709-2}
y(s_{ll}^{u_{j'}}(g))\leq y(s_{ll}^{u_i}(g)).
\end{equation}
By the definition of configuration,
\begin{equation}\label{eq0709-1}
y(s)=y(s_{hl}^{u_i}(g))\geq y(s_{ll}^{u_i}(g)).
\end{equation}
In fact, since we have assumed that squares are in generic positions and are now considering the case $s_{ll}^{u_i}(g)\neq s$, inequality \eqref{eq0709-1} must be strict. Combining this with inequality \eqref{eq0709-2}, we have $y(s_{ll}^{u_{j'}}(g))<y(s)$ and thus $s_{ll}^{u_{j'}}(g)\neq s$.

{\bf Case 2:} $s$ appears in $S^{u_i}$ as $s_{next}$.

Note that $s$ disappears from $S^{u_j}$ implies that line $\ell^{left}_g$ has stridden over the left boundary of $s=s^{u_i}_{next}(g)$. Hence $x(s)<x(\ell^{right}_{u_j,g})\leq x(\ell^{right}_{u_{j'},g})$. Then by the definition of configuration, we have
$$
x(s)<x(\ell^{right}_{u_{j'},g})\leq \min\{x(s_{next}^{u_{j'}}),x(s_{hr}^{u_{j'}}(g)),x(s_{lr}^{u_{j'}}(g))\},
$$
and thus $s$ will never serve as $s_{next}$, $s_{hr}$ or $s_{lr}$ after $u_j$.

Next, we show that
\begin{equation}\label{eq0714-1}
\mbox{$s$ will not serve as $s_{hl}$ or $s_{ll}$ after $u_j$, too.}
\end{equation}
Because $s$ disappears from $S^{u_j}$, there must exist an index $z$ with $i\leq z<j$ such that $s_{next}^{u_z}(g)=s$ and $s_{next}^{u_{z+1}}(g)\neq s$. At this time, $g$ serves as $g'$ in $(\romannumeral1)$ of the DAG, and thus the third condition of $(\romannumeral2_3)$ is satisfied taking $u=u_z$ and $v=u_{z+1}$. It follows that
\begin{align*}
& y(s_{ll}^{u_{j'}}(g))\leq y(s_{ll}^{u_j}(g))\leq y(s_{ll}^{u_{z+1}}(g))\leq y(s_{next}^{u_z}(g))=y(s)\ \mbox{and}\\
& y(s)=y(s_{next}^{u_z}(g))\leq y(s_{hl}^{u_{z+1}}(g))\leq y(s_{hl}^{u_j}(g)) \leq y(s_{hl}^{u_{j'}}(g)).
\end{align*}
Note that for each of the above sequences of inequalities, there must be a strict inequality, because $s$ has disappeared from $S^{u_j}$ and we have assume that squares are in generic positions. So, \eqref{eq0714-1} is proved.

{\bf Case 3:} $s$ appears in $S^{u_i}$ as $s_{hr}$ or $s_{lr}$.

Suppose $s=s_{hr}^{u_i}(g)$ (the case when $s=s_{lr}^{u_i}(g)$ is similar). After $u_j$, $s$ will never serve as $s_{lr}$ or $s_{hr}$, because $y(s_{lr}^{u_{j'}}(g))\leq y(s_{hr}^{u_{j'}}(g))<y(s_{hr}^{u_i}(g))$. Recall that $S^{u_{j-1}}$ is the last configuration of the maximal consecutive segment containing $s$. If $s$ appears in $S^{u_{j-1}}$ as $s_{hl}$, $s_{ll}$ or $s_{next}$, then similar to Case 1 and Case 2, $s$ will never reappear after $S^{u_j}$. Next consider the case when $s=s_{hr}^{u_{j-1}}(g)$. In this case, by the second condition of $(\romannumeral2_2)$ in the definition of DAG and because $s$ disappears from $S^{u_j}$, we must have
\begin{equation}\label{eq0716-1}
 x(\ell^{right}_{u_j,g})>x(s_{hr}^{u_{j-1}}(g))=x(s).
\end{equation}
By the definition of configuration and the monotonicity of sweep lines, $x(s_{next}^{u_{j'}})(g)> x(\ell^{right}_{u_{j'},g})\geq x(\ell^{right}_{u_j,g})$. Hence $x(s_{next}^{u_{j'}})(g)>x(s)$, and thus $s$ will not appear as $s_{next}$ after $u_j$. Furthermore, combining \eqref{eq0716-1} with the second condition of $(\romannumeral2_3)$, and because $s$ disappears from $u_j$, we have
\begin{equation}\label{eq0715}
y(s_{hl}^{u_{j'}}(g))\geq y(s_{hl}^{u_j}(g))>y(s_{hr}^{u_{j-1}}(g))=y(s)>y(s_{ll}^{u_j}(g))\geq y(s_{ll}^{u_{j'}}(g)).
\end{equation}
So, $s$ will not serve as $s_{hl}$ or $s_{ll}$ after $u_j$. Similar argument shows the same conclusion if $s=s_{hl}^{u_{j-1}}(g)$. To sum up all the above cases, $s$ is not double-counted in Case 3.

From above three cases, any $s\in \mathcal S(P)$ is counted exactly once in $w(P)$ and the righthand side inequality is proved.
\end{proof}
The next lemma shows that any feasible solution to the D$k$SH-US instance corresponds to a source-sink path in $G$ whose weight equals the cost of the solution.

\begin{lemma}\label{lem0201-1}
Let $\mathcal O^*$ be a feasible solution to the D$k$SH-US$_b$ instance. Then there is a source-sink path $Q$ in $G_{b,k_b}$ with $w(Q)=|\mathcal O^*|$.
\end{lemma}
\begin{proof}
The desired path $Q$ can be constructed by tracing $\{\mathcal O^*_g\}$ simultaneously. For each square $s\in \mathcal O^*$, let $x'(s)=x(s)-x(g_s)-\varepsilon$, where $g_s$ is the grid point contained in $s$. Order the unit squares in $\mathcal O^*$ as $s_1,s_2,\ldots,s_T$ such that $x'(s_1)<x'(s_2)<\cdots <x'(s_T)$. For each grid point $g$, order unit squares in $\mathcal O^*_g$ as $\{s_{g,1},\ldots,s_{g,t_{g}}\}$, where $x(s_{g,1})<\cdots x(s_{g,t_{g}})$.

We construct a source-sink path $Q=u_{src}u_0u_1\ldots u_Tu_{sink}$ as follows, where every $u_j$ has the form $(S^{u_j},\tilde k^{u_j},x_{u_j})$ and $x_{u_j}=x'(s_j)$ for $j=0,1,\ldots,T$ (the idea is to move the sweep line step by step along the positions of $s_1,\ldots,s_T$), where $x'(s_0)=0$. Referring to Table \ref{tab1} for the example in Fig. \ref{fig0229-1} might be helpful in understanding the construction.

Let $u_0=(S^{u_0},\tilde k^{u_0},0)$ be the configuration with the following structure: every 5-tuple in $S^{u_0}$ has the form $(s^{b}_g, s^{b}_g, s_{hr}(g),s_{lr}(g);s_{next}(g))$, where $s_{next}(g)=s_{g,1}$, $s_{hr}(g)$ and $s_{lr}(g)$ are the highest and the lowest squares of $\mathcal O^*_g$ cut by line $\ell(x(g)-\varepsilon)$; $\tilde k^{u_0}$ is the number of points covered by $U(S^{u_0})$. By the construction of DAG, $(u_{src},u_0)$ is an arc in $G$. Suppose by induction that we have found a path $u_{src}u_0\ldots u_{j-1}$ in $G_{b,k_b}$. Construct a configuration $u_j=(S^{u_j},\tilde k^{u_j},x_{u_j})$ in the following way: let $g'$ be the grid point contained in $s_{j+1}$; all 5-tuples in $S^{u_j}$ remain the same except for $S^{u_j}(g')$, in which $s_{next}(g')=s_{j+1}$, $s_{hl}(g')$ and $s_{ll}(g')$ are the highest and the lowest squares of $\mathcal O^*_g$ cut by line $\ell(x(s_{j+1})-1-\varepsilon)$, $s_{hr}(g')$ and $s_{lr}(g')$ are the highest and the lowest squares of $\mathcal O^*_g$ cut by line $\ell(x(s_{j+1})-\varepsilon)$ (using $s_g^e$ if the above $s_{hr}(g')$ and $s_{lr}(g')$ do not exist); $\tilde k^{u_j}$ equals the sum of $\tilde k^{u_{j-1}}$ and the number of points covered by $U(S^{u_{j}})\setminus U(S^{u_{j-1}})$. By Definitions \ref{def0131}, Definition \ref{def0131-1}, Remark \ref{rem0719-3} and Remark \ref{rem0719-6}, $u_j$ is a valid configuration and $(u_{j-1},u_j)$ is an arc in $G_{b,k_b}$. Continuing in this way, we could find a path $u_{src}u_0u_1\ldots u_T$, where vertex $u_T$ corresponds to a configuration $(S^{u_T},\tilde k^{u_T},1)$ with the following form: every 5-tuple in $S^{u_T}$ has the form $(s_{hl}(g),s_{ll}(g),s_{g}^e,s_{g}^e;s_{g}^e)$, where $s_{hl}(g)$ and $s_{ll}(g)$ are the highest and the lowest squares of $\mathcal O^*_g$ cut by line $\ell(x(g)-\varepsilon)$. Denote $\tilde Q=u_{src}u_0u_1\ldots u_T$.

To finish the construction of $Q$, what remains to show is that
\begin{equation}\label{eq0113}
\tilde{k}^{u_T}=k_b\ \mbox{and thus $u_T$ is linked to $u_{sink}$.}
\end{equation}
For this purpose, it suffices to prove that
\begin{equation}\label{eq0618-1}
\mbox{every point covered by $\mathcal O^*$ is counted exactly once in the accumulation of $\tilde k^{u_T}$.}
\end{equation}

For a point $p$ covered by $\mathcal O^*$, among the squares in $\mathcal O^*$ containing $p$,
\begin{equation}\label{eq0720-3}
\mbox{denote by $s_{t}$ the square with the largest $(x(s)-x(g_s))$-value,}
\end{equation}
where $g_s$ is the grid point contained in $s$. We say that vertex $u$ covers point $p$ if $p$ is contained in some square of $S^u$. Note that when going along $\tilde Q$, every time we meet a {\em maximal consecutive} segment covering $p$, point $p$ is counted once. So, to prove \eqref{eq0618-1}, what we need to show is that there is only one maximal consecutive segment of $\tilde Q$ covering $p$. For this purpose, we first prove the following claim.

\vskip 0.2cm {\bf Claim.} Let $u_j$ be the first vertex of $\tilde Q$ covering $p$, and let $u_{j'-1}$ be the end of this maximal consecutive segment covering $p$. Then $x(\ell_{u_{j'},g_{s_t}}^{right})>x(s_{t})$.

We only consider the case when $p$ is on the right of $g_{s_t}$ (the argument for the case when $p$ is on the left of $g_{s_t}$ is similar). We prove the claim by contradiction. Note that $x(\ell_{u_{j'},g_{s_t}}^{right})\leq x(s_{t})$ implies that the union $s_{hr}^{u_{j'}}(g_{s_t})\cup s_{lr}^{u_{j'}}(g_{s_t})$ contains the area of $s_t$ between line $\ell(x(g_{s_t}))$ and line $\ell_{u_{j'},g_{s_t}}^{right}$. So, if $p$ is on the left-side of $\ell_{u_{j'},g_{s_t}}^{right}$, then either $s_{hr}^{u_{j'}}(g_{s_t})$ or $s_{lr}^{u_{j'}}(g_{s_t})$ contains $p$, contradicting that $u_{j'}$ does not cover $p$. Hence
\begin{equation}\label{eq0605-1}
\mbox{if the claim is not true, then}\ x(p)>x(\ell_{u_{j'},g_{s_t}}^{right}).
\end{equation}

Let $s_a$ be a square of $S^{u_j}$ containing $p$ and assume that $s_a$ is associated with grid point $g_{s_a}$. Since $p$ belongs to both $s_a$ and $s_t$, squares $s_a$ and $s_t$ intersect. Hence $x(g_{s_t})$ can only be $x(g_{s_a})-1$ or $x(g_{s_a})$ or $x(g_{s_a})+1$.  If $x(g_{s_t})=x(g_{s_a})+1$, then $x(s_a)+1>x(s_t)$ (because $s_a$ and $s_t$ intersect), which contradicts the assumption that $s_t$ has the largest $(x(s)-x(g_s))$-value among all squares of $\mathcal O^*$ containing $p$. So there are two cases left.

{\bf Case 1.} $x(g_{s_a})=x(g_{s_t})$.

To obtain a contradiction in this case, we first prove
\begin{equation}\label{eq0531-1}
x(\ell_{u_{j'},g_{s_a}}^{right})>x(s_a).
\end{equation}
In fact, if \eqref{eq0531-1} is not true, then $\ell_{u_{j'},g_{s_a}}^{left}$ is on the left of $s_a$, and thus $s_a$ cannot be $s^{u_j}_{ll}(g_{s_a})$ or $s^{u_j}_{hl}(g_{s_a})$.
Suppose $s_a$ appears as $s^{u_j}_{lr}(g_{s_a})$. By the monotonicity of envelope, as long as $\ell^{right}_{g_{s_a}}$ has not stridden over the right boundary of $s_a$, then $s_a$ remains to be $s_{lr}(g_{s_a})$ in the following configurations of $Q$. Combining this observation with the assumption that $u_{j'}$ does not cover $p$, and thus $S^{u_{j'}}$ does not contain $s_a$, we have $x(\ell_{u_{j'},g_{s_a}}^{right})>x(s_a)$. A similar argument shows the validity of \eqref{eq0531-1} in the case when $s_a$ appears as $s^{u_j}_{hr}(g_{s_a})$ or $s^{u_j}_{next}(g_{s_a})$.

Combining \eqref{eq0605-1} and \eqref{eq0531-1}, making use of $g_{s_t}=g_{s_a}$, we have $x(p)>x(s_a)$, contradicting the assumption that square $s_a$ contains point $p$. So, the claim holds in this case.


%

{\bf Case 2:} $x(g_{s_a})=x(g_{s_t})+1$.

In this case,
\begin{equation}\label{eq0720-1}
\mbox{point $p$ is on the left side of grid point $g_{s_a}$,}
\end{equation}
and $x(\ell^{right}_{u_{j'},g_{s_t}})=x(\ell^{left}_{u_{j'},g_{s_a}})$. Then by \eqref{eq0605-1},
\begin{equation}\label{eq0617-1}
\mbox{if the claim is not ture, then $p$ is on the right side of $\ell^{left}_{u_{j'},g_{s_a}}$.}
\end{equation}
Note that the union $s_{hl}^{u_{j'}}(g_{s_a})\cup s_{ll}^{u_{j'}}(g_{s_a})\cup s_{next}^{u_{j'}}(g_{s_a})$ contains the area of $s_a$ between the line $\ell_{u_{j'},g_{s_a}}^{left}$ and the line $\ell(x(g_{s_a}))$. Combining this observation  with \eqref{eq0720-1} and \eqref{eq0617-1}, if the claim is not true, then $s_{hl}^{u_{j'}}(g_{s_a})\cup s_{ll}^{u_{j'}}(g_{s_a}) \cup s_{next}^{u_{j'}}(g_{s_a})$ contains $p$, contradicting the assumption that $u_{j'}$ does not cover $p$. In any case, we have proved the claim.

Combining the claim with the construction of the DAG and the definition of $s_t$ in \eqref{eq0720-3}, for any configuration $u_i$ with $i\geq j'$ and any square $s\in \mathcal O^*$ containing $p$, we have
$$
x(\ell_{u_{j'},g_s}^{right})>x(s),
$$
and the following two observations hold.

$(a)$ $x(\ell_{u_i,g_s}^{right})\geq x(\ell_{u_{j'},g_s}^{right})>x(s)$. By the definition of configuration, $x(\ell_{u_i,g_s}^{right}) < \min\{x(s^{u_i}_{hr}(g_s)),$ $x(s^{u_i}_{lr}(g_s)),x(s^{u_i}_{next}(g_s))\}$, and thus $s$ cannot appear as $s_{hr},s_{lr},s_{next}$ in $S^{u_i}$.

$(b)$ $x(\ell_{u_i,g_s}^{left})\geq x(\ell_{u_{j'},g_s}^{left})>x(s)-1$. Since $s\not\in S^{u_{j'}}(g_s)$, it is neither the highest square nor the lowest square at $\ell_{u_{j'},g_s}^{left}$. By the monotonicity of $y$-coordinate for the DAG, $s$ cannot appear as $s_{hl}$ or $s_{ll}$ in $S^{u_i}$.

By these observations, no $u_i$ with $i\geq j'$ can cover $p$. Then property \eqref{eq0618-1} is proved, and \eqref{eq0113} follows. This finishes the construction of a source-sink path $Q$ in $G_{b,k_b}$.

Note that $w(Q)=|\mathcal S(Q)|$ because of Lemma \ref{lem0217}. By the construction of $Q$, we have $\mathcal S(Q)= \mathcal O^*$ (note that every $s\in\mathcal O^*$ has served as a $s_{next}$ in the construction, and the construction only uses squares in $\mathcal O^*$). So, $w(Q)=\mathcal |O^*|$.
\end{proof}

Combining Lemma \ref{lem0217} and Lemma \ref{lem0201-1}, we have the following result

\begin{theorem}\label{thm0115}
A maximum weight source-sink path $P^*$ in the auxiliary digraph $G_{b,k_b}$ satisfies $|\mathcal S(P^*)|\geq |\mathcal O^*|$, where $\mathcal O^*$ is any feasible solution to the D$k$SH-US$_b$ instance. Furthermore, $\mathcal S(P^*)$ covers at most $k_b$ points. As a consequence, an optimal solution to a D$k$SH-US$_b$ instance can be found by finding a maximum-weight source-sink path in $G_{b,k_b}$, which can be done in time $O(k_b^2m_b^{O(a^2)})$, where $m_b=|\mathcal S_b|$.
\end{theorem}
\begin{proof}
By Lemma \ref{lem0217}, $|\mathcal S(P^*)| =w(P^*)$. By Lemma \ref{lem0201-1}, the constructed source-sink path $Q$ has $w(Q)=|\mathcal O^*|$. Since $P^*$ is a maximum weight source-sink path, we have $w(P^*)\geq w(Q)$. Hence $|\mathcal S(P^*)|\geq |\mathcal O^*|$. By Lemma \ref{lem0217}, $\mathcal S(P^*)$ covers at most $k_b$ points.

Next, we consider the time complexity. Note that there are $O(k_bm_b^{O(a^2)})$ configurations. This is because for a configuration $(S,\tilde k,x)$, there are $k_b+1$ choices for $\tilde k$; for each grid point $g$, $S(g)$ has 5 token squares, and there are $(a+1)^2$ grid points, so there are $m_b^{5(a+1)^2}$ choices for $S$; since the positions can be discretized by the right boundaries of the unit squares, $x$ has $m_b+2$ choices (including the beginning and the ending positions). So, the graph $G_{b,k_b}$ has $O(k_bm_b^{O(a^2)})$ vertices and $O(k_b^2m_b^{O(a^2)})$ edges. Then, the time complexity follows from the observation that a maximum-weight source-sink path in a DAG $G=(V,E)$ can be computed in time $O(|V|+|E|)$ (see for example \cite{Cormen}).
\end{proof}

\subsection{Assembling Local Solutions}\label{sec2.3}

The algorithm for the original region implements the shifting and partition technique which was first proposed by Hochbaum and Maass in \cite{Hochbaum85}. Its pseudocode is given in Algorithm \ref{algo1}, and the notations used in the algorithm are described below.

\begin{algorithm}[H]
\caption{Algorithm for D$k$SH-US}
{\bf Input:} An instance of D$k$SH-US $(V,E)$.
	
{\bf Output:} A set of points $V'\subseteq V$.
\begin{algorithmic}[1]
\State Form partitions $P_0,\ldots,P_{a-1}$.
\For {each integer $r\in \{0,\ldots,a-1\}$}
\State Compute $\{B_i^{K_i}\}_{i=0,1,\ldots,t}$ using transition formula \eqref{eq0630-1} and \eqref{eq0625-1}.
\State $\mathcal S_r\leftarrow \arg\max\{c(B^{K_t}_t)\colon K_t\leq 4k\}$ \label{line06628-1}
\EndFor
\State $\mathcal S'\leftarrow \arg\max_{r\in\{0,1,\ldots,a-1\}}\{c(\mathcal S_r)\}$
\State Output $V'\leftarrow$ the set of points covered by $\mathcal S'$.
\end{algorithmic}\label{algo1}
\end{algorithm}

Assume that $A$ is a square containing all the $n$ points of $V$, which has size $a_0\times a_0$. For an integer $a$ which will be determined later, let $N=\lceil a_0/a\rceil$. Extend $A$ into square
$$A_0=\{(x,y)\colon -a\leq x\leq Na, -a\leq y\leq Na\}.$$
Partition $A_0$ into $(N+1)^2$ {\em blocks} of size $a\times a$. Denote this partition of $A_0$ as $P_0$. Note that the lower-left corner of $P_0$ is $(-a,-a)$. For each integer $r \in \{0,\ldots,a-1\}$, construct a partition $P_r$ for square $A_r=\{(x,y)\colon -a+r\leq x\leq Na+r, -a+r\leq y\leq Na+r\}$ by shifting $P_0$ to the upper-right direction by a vector $(r,r)$. Note that every $A_r$ contains all the points of $V$.
The algorithm will solve the problem for each partition, and then picks the best one, where the meaning of ``best'' is in line \ref{line06628-1} of the algorithm.

Note that a solution to the D$k$SH-US instance is point set $V'$. For simplicity of statement, we say that a unit square $s$ is contained in $V'$ if the hyperedge corresponding to $s$ is contained in $V'$, and abuse terminology a little by calling the set of unit squares contained in $V'$ also as a solution.

To compute the problem for a partition $P_r$, we have to {\em guess} the number of points used in each block. This can be done by a dynamic programming method. Order the nontrivial blocks as $b_1,b_2,\ldots,b_t$ (a block is nontrivial if it contains some point) and then process them sequentially. The idea is to iteratively compute a solution to the D$k$SH-US instance confined to the first $i$ blocks, that is, a set $B_i^{K_i}$ of unit squares intersecting the first $i$ blocks covering at most $K_i$ points. However, we could not do it. Instead, making use of the method in the above subsection, we calculate an approximate $B_i^{K_i}$ in the following way.
For the $j$th block and a nonnegative integer $k_j$, let $P_{j,k_j}$ be a maximum-weight source sink path in $G_{j,k_j}$, where $G_{j,k_j}$ is the graph constructed in Definition \ref{def0131-1}. Recall that if $P_{j,k_j}=u_{src}u_0u_1\ldots u_Tu_{sink}$, then $\mathcal S(P_{j,k_j})=\bigcup_{i=0}^TU(S^{u_i})$. Note that if such a path does not exist, then $P_{j,k_j}=null$ and the cost of $null$ is $-\infty$. Let
\begin{equation}\label{eq0630-2}
B_{i}^{K_i}=\arg\max\{\sum_{j=1}^i|\mathcal S(P_{j,k_j})|\colon \sum_{j=1}^{i}k_j=K_i\}.
\end{equation}
It should be emphasized that $B_i^{K_i}$ is a {\em multi-set}, some unit squares striding over more than one blocks might be counted more than once. Define the cost of $B_{i}^{K_i}$ to be $c(B_{i}^{K_i})=\sum_{j=1}^i|\mathcal S(P_{j,k_j})|$. The transition formula for $c(B_i^{K_i})$ is as follows.
\begin{equation}\label{eq0630-1}
c(B_0^{K_0})=\left\{\begin{array}{ll}0, & K_0=0. \\
-\infty, & \mbox{otherwise},\end{array}\right.
\end{equation}
and for $i=1,2,\ldots,t$,

\begin{equation}\label{eq0625-1}
c(B^{K_i}_{i})=\max_{K_{i-1}\colon 0\leq K_{i-1}\leq K_i}\{|\mathcal S(P_{i,K_i-K_{i-1}})|+ c(B^{K_{i-1}}_{i-1})\}.
\end{equation}

The output of the algorithm is the best $B_t^{K_t}$ covering at most $4k$ points, over all partitions.
The following theorem shows the approximation ratio of the algorithm.

\begin{theorem}\label{thm0117}
Setting $a=\lceil 3/\varepsilon\rceil$, Algorithm \ref{algo1} computes a $(4,\frac{1}{1+\varepsilon})$-bicriteria approximate solution to the D$k$SH-US instance in time $O(\frac{1}{\varepsilon} n^3m^{O(1/\varepsilon^2)})$, where $n$ and $m$ are the number of points and the number of unit squares, respectively.
\end{theorem}
\begin{proof}
By line \ref{line06628-1} of the algorithm and the definition of $B_i^{K_i}$, the set $\mathcal S'$ computed by Algorithm \ref{algo1} covers at most $4k$ points.

Let $V^*$ be an optimal solution to the D$k$SH-US instance, and let $\mathcal O^*$ be the set of unit squares contained in $V^*$. Consider a partition $P_r$. For block $b$ in $P_r$, let $\mathcal O^*_{r,b}$ be the set of unit squares in $\mathcal O^*$ which have nonempty intersections with $b$. Suppose $\mathcal O^*_{r,b}$ covers $k_b$ points. Let $P^*_{r,b,k_b}$ be a maximum-weight source-sink path in $G_{r,b,k_b}$, where $G_{r,b,k_b}$ is the DAG in Definition \ref{def0131-1} with respect to partition $P_r$. By Theorem \ref{thm0115},
\begin{equation}\label{eq0115-1}
\mathcal S(P^*_{r,b,k_b})\ \mbox{covers at most $k_b$ points and}\ |\mathcal S(P^*_{r,b,k_b})|\geq |\mathcal O^*_{r,b}|.
\end{equation}
It follows that $\bigcup_{\scriptsize\mbox{block}\ b\ \mbox{of} \ P_r}\mathcal S(P^*_{r,b,k_b})$ covers at most $\sum_{\scriptsize\mbox{block}\ b\ \mbox{of}\ P_r}k_b$ points. Note that $\sum_{\scriptsize\mbox{block}\ b\ \mbox{of}\ P_r}k_b\leq 4k$,
because every point covered by $\mathcal O^*$ is counted at most four times in the left term (the repetition comes from unit squares intersecting the boundaries of the blocks). Hence $\bigcup_{\scriptsize\mbox{block}\ b\ \mbox{of}\ P_r}\mathcal S(P^*_{r,b,k_b})$ is a candidate choice of $B_t^{K_t}$ for some $K_t\leq 4k$). By the definition of $B_i^{K_i}$ in \eqref{eq0630-2}, the set $\mathcal S_r$ computed in line \ref{line06628-1} of the algorithm and the final output $\mathcal S'$ of the algorithm satisfy
\begin{equation}\label{eq0622-6}
c(\mathcal S')\geq c(\mathcal S_r)\geq\sum_{\scriptsize\mbox{block}\ b\ \mbox{of}\ P_r}|\mathcal S(P^*_{r,b,k_b})|,\ \mbox{for any}\ r\in\{0,1,\ldots,a-1\}.
\end{equation}

Let $\mathcal H_r$ (resp. $\mathcal V_r$) be the set of unit squares in $\mathcal S'$ which intersect some horizontal (resp. vertical) grid lines of $P_r$. Observe that a unit square can intersect at most four blocks of $P_r$. In other words, the multiplicity of each unit square in $\mathcal S'$ is at most four. Hence,
\begin{equation}\label{eq0622-8}
c(\mathcal S')\leq |\mathcal S'|+|\mathcal H_r|+2|\mathcal V_r|.
\end{equation}
Note that a unit square can not be in both $\mathcal H_r$ and $\mathcal H_{r'}$ for $r\neq r'$. Hence
\begin{equation}\label{eq0622-9}
\sum\limits_{r=0}^{a-1}|\mathcal H_r|\leq |\mathcal S'|.
\end{equation}
Similarly,
\begin{equation}\label{eq0622-10}
\sum\limits_{r=0}^{a-1}|\mathcal V_r|\leq |\mathcal S'|.
\end{equation}

Combining inequalities \eqref{eq0115-1} to \eqref{eq0622-10}, we have
$$
a\cdot |\mathcal O^*|\leq \sum\limits_{r=0}^{a-1}\sum_{\scriptsize\mbox{block}\ b\ \mbox{of}\ P_r}|\mathcal O^*_{c,b}|\leq a\cdot |\mathcal S'|+\sum_{r=0}^{a-1}|\mathcal H_r|+2\sum_{r=0}^{a-1}|\mathcal V_r|\leq (a+3)|\mathcal S'|.
$$
It follows that when $a=\lceil 3/\varepsilon\rceil$,
$$
|\mathcal S'|\geq
\frac{a}{(a+3)}|\mathcal O^*|\geq \frac{1}{1+\varepsilon}|\mathcal O^*|.
$$
The approximation ratio is proved.

Next, consider the time complexity. By Theorem \ref{thm0115}, for each block $b$ of each partition $P_r$, and each integer $k_b$, computing a maximum-weight source-sink path needs time $O(k_b^2m_b^{O(a^2)})=O(n^2m^{O(a^2)})$. Note that the transition formula \eqref{eq0625-1} needs $k_b$ to vary in $\{0,1,\ldots,n\}$. Computing all these maximum-weight source-sink paths needs time $O(n^3m^{O(a^2)})$. For each partition $P_r$, there are $O(n^2)$ sets of the form $B_i^{K_i}$ to be computed, because for each $i=0,1,\ldots,t$ ($t\leq n$), there are $\min\{n+1,4k\}$ choices for $K_i$. Hence, using the transition formula \eqref{eq0625-1} to compute a $B_i^{K_i}$, with all maximum-weight source-sinks paths at hand, needs time $O(n^3)$. So, $\{B_t^{K_t}\}$ can be obtained in time $O(n^3m^{O(a^2)})+O(n^3)=O(n^3m^{O(a^2)})$. Since there are $a=\lceil 3/\varepsilon\rceil$ partitions, the time complexity is $O(an^3m^{O(a^2)})=O(\frac{1}{\varepsilon} n^3m^{O(1/\varepsilon^2)})$.
\end{proof}

Combining Theorem \ref{thm1226} and Theorem \ref{thm0117}, we have the following result, where the additional factor $n$ in the time complexity comes from the proof of Theorem \ref{thm1226} that we have to call the D$k$SH-US algorithm for $k=1,\ldots,|V|$.
\begin{theorem}
There exits a $(\frac{1}{1+\varepsilon},4)$-bicriteria approximation algorithm for Min$p$U-US in time $O(\frac{1}{\varepsilon} n^4m^{O(1/\varepsilon^2)})$.
\end{theorem}

\section{Conclusion and Discussion}\label{sec3}

In this paper, for the unit-square minimum $p$-union problem (Min$p$U-US), we designed a $(\frac{1}{1+\varepsilon},4)$-bicriteria approximation algorithm, which exploits a relation between Min$p$U and the densest $k$-subhypergraph problem (D$k$SH), and makes full use of geometric structures of unit-squares. New techniques have to be explored in order to further reduce the approximation factor or to obtain a non-bicriteria approximation algorithm.

\section*{Acknowledgment}
This research is supported in part by National Natural Science Foundation of China (11901533, U20A2068, 11771013), and Zhejiang Provincial Natural Science Foundation of China (LD19A010001).

\end{document}